\documentclass[10pt,journal]{IEEEtran} 
\usepackage{url}
\usepackage{tablefootnote}
\usepackage{longtable}
\usepackage{xspace}
\usepackage{comment}
\usepackage{subfigure}
\usepackage{multirow}
\usepackage{hyperref}
\usepackage{cleveref}
\usepackage{amsmath}
\usepackage{algorithmicx}
\usepackage{algorithm}
\usepackage{algpseudocode}
\usepackage{graphicx}
\usepackage{amssymb,amsmath,amsthm, mathtools}
\newtheorem{theorem}{Theorem}

\newtheorem{definition}{Definition}

\DeclareMathOperator*{\argmax}{\arg\!\max}
\DeclareMathOperator{\bwratio}{BWRatio}

\newcommand{\packetsize}{{\sc\textsc{packet-size}}\xspace}
\newcommand{\initialrho}{{\sc\textsc{initial-rho}}\xspace}
\newcommand{\quiettime}{{\sc\textsc{quiet-time}}\xspace}

\newcommand{\buflo} {BuFLO\xspace}
\newcommand{\csbuflo} {Congestion-Sensitive BuFLO\xspace}
\newcommand{\csb} {CS-BuFLO\xspace}

\newcount\Comments  
\usepackage{color}
\newcommand{\kibitz}[2]{\ifnum\Comments=1\textcolor{#1}{#2}\fi}

\begin{document}

\title{New Approaches to Website Fingerprinting Defenses}

\author{Xiang Cai, Rishab Nithyanand, Rob Johnson\\Stony Brook University}
\date{}

\maketitle

\begin{abstract}
  Website fingerprinting attacks\cite{hintz-pets02} enable an
  adversary to infer which website a victim is visiting, even if the
  victim uses an encrypting proxy, such as Tor\cite{tor-website}.
  Previous work has shown that all proposed defenses against website
  fingerprinting attacks are ineffective\cite{dyer-snp12,cai-ccs12}.
  This paper advances the study of website fingerprinting attacks and
  defenses in two ways.  First, we develop bounds on the trade-off
  between security and bandwidth overhead that any fingerprinting
  defense scheme can achieve.  This enables us to compare schemes with
  different security/overhead trade-offs by comparing how close they
  are to the lower bound.  We then refine, implement, and evaluate the
  \csbuflo scheme outlined by Cai, et al.~\cite{cai-ccs12}.  \csb,
  which is based on the provably-secure \buflo defense proposed by
  Dyer, et al.\cite{dyer-snp12}, was not fully-specified by Cai, et
  al, but has nonetheless attracted the attention of the Tor
  developers~\cite{perry-critique,perry-tbdesign}.  Our experiments
  find that \csbuflo has high overhead (around 2.3-2.8x) but can get
  6$\times$ closer to the bandwidth/security trade-off lower bound
  than Tor or plain SSH.
\end{abstract}


\section{Introduction}
\label{sec:introduction}

Website fingerprinting attacks have emerged as a serious threat
against web browsing privacy mechanisms, such as SSL, Tor, and
encrypting tunnels.  These privacy mechanisms encrypt the content
transferred between the web server and client, but they do not
effectively hide the size, timing, and direction of packets.  A
website fingerprinting attack uses these features to infer the web
page being loaded by a client.

Researchers have engaged in a war of escalation in developing website
fingerprinting attacks and defenses, with two recent papers
demonstrating that all previously-proposed defenses provide little
security\cite{dyer-snp12,cai-ccs12}.  At the 2012 Oakland conference,
Dyer, et al. showed that an attacker could infer, with a success rate
over 80\%, which of 128 pages a victim was visiting, even if the
victim used network-level countermeasures.  They also performed a
simulation-based evaluation of a hypothetical defense, which they call
\buflo, and found that it required over 400\% bandwidth overhead in
order to reduce the success rate of the best attack to 5\%, which is
still well-above the ideal 0.7\% success rate from random guessing.
At CCS 2012, Cai et al. proposed the DLSVM fingerprinting attack and
demonstrated that it could achieve a greater than 75\% success rate
against numerous defenses\cite{cai-ccs12}, including application-level
defenses, such as HTTPOS\cite{luo-ndss11} and randomized
pipelining\cite{tor-randomized-pipelining}.  As a result, it is not
currently known whether there exists any efficient and secure defense
against website fingerprinting attacks.

Cai, et al. also proposed \csbuflo, which extended Dyer's \buflo
scheme to include congestion sensitivity and some rate adaptation, but
they left many details unspecified and did not implement or evaluate
their scheme.  Despite the lack of data on \csb, the Tor project has
indicated interest in incorporating \csb into the Tor
browser~\cite{perry-tbdesign,perry-critique}.

In order to get a better understanding of the performance and security
of the \csb protocol, this paper presents a complete specification of
\csb, describes an SSH-based implementation, and evaluates its
bandwidth overhead, latency overhead, and security against the current
best-known attacks.

Cai's description of the \csb protocol outlines solutions to several
performance and practicality problems in the original \buflo protocol
-- \csb is TCP-friendly, it pads streams in a uniform way, and it uses
information collected offline to tune \buflo's parameters to the
website being loaded.  We propose two further improvements: we modify
\csb to adapt its transmission rate dynamically, and we improve its
stream padding to use less bandwidth while hiding more information
about the website being loaded.  Dynamic rate adaptation makes \csb
much more practical to deploy, since it does not require an
infrastructure for performing offline collection of statistics about
websites, but poses a challenge: adapting too quickly to the website's
transmission rate can reveal information about which website the
victim is visiting.  \csb balances these performance and security
constraints by limiting the rate and precision of adaptation.



\begin{table*}
\begin{minipage}{\textwidth}                                                                                         
  \begin{center}
  \begin{tabular}{|lrllrrrrr|}
    \hline
    Defense                            & $n$ & Method     & Source                  & Panchenko & VNG++ & DLSVM & BW Ratio & Latency Ratio \\
    \hline
    \hline
    \csb (CTSP)                        & 200 & Empirical  & this paper              & 18.0      & 13.0  & 20.6  & 2.796    & 3.271         \\
    \csb (CPSP)                        & 200 & Empirical  & this paper              & 24.2      & 16.5  & 34.3  & 2.289    & 2.708         \\
    \csb (CTSP)                        & 120 & Empirical  & this paper              & 23.4      & 20.9  & 28.9  & 2.799    & 3.444         \\
    \csb (CPSP)                        & 120 & Empirical  & this paper              & 30.6      & 22.5  & 40.5  & 2.300    & 2.733         \\
    \hline
    \buflo($\tau=0,\rho=40,d=1000$)    & 128 & Simulation & \cite{dyer-snp12}       & 27.3      & 22.0  & N/A   & 1.935    & N/A           \\
    \buflo($\tau=0,\rho=40,d=1500$)    & 128 & Simulation & \cite{dyer-snp12}       & 23.3      & 18.3  & N/A   & 2.200    & N/A           \\
    \buflo($\tau=0,\rho=20,d=1000$)    & 128 & Simulation & \cite{dyer-snp12}       & 20.9      & 15.6  & N/A   & 2.405    & N/A           \\
    \buflo($\tau=0,\rho=20,d=1500$)    & 128 & Simulation & \cite{dyer-snp12}       & 24.1      & 18.4  & N/A   & 3.013    & N/A           \\
    \buflo($\tau=10^5,\rho=40,d=1000$) & 128 & Simulation & \cite{dyer-snp12}       & 14.1      & 12.5  & N/A   & 2.292    & N/A           \\
    \buflo($\tau=10^5,\rho=40,d=1500$) & 128 & Simulation & \cite{dyer-snp12}       & 9.4       & 8.2   & N/A   & 2.975    & N/A           \\
    \buflo($\tau=10^5,\rho=20,d=1000$) & 128 & Simulation & \cite{dyer-snp12}       & 7.3       & 5.9   & N/A   & 4.645    & N/A           \\
    \buflo($\tau=10^5,\rho=20,d=1500$) & 128 & Simulation & \cite{dyer-snp12}       & 5.1       & 4.1   & N/A   & 5.188    & N/A           \\
    \hline
    HTTPOS                             & 100 & Empirical  & \cite{cai-ccs12}        & 57.4      & N/A   & 75.8  & 1.361    & N/A           \\
    Tor+rand. pipe.                    & 100 & Empirical  & \cite{cai-ccs12}        & 62.8      & N/A   & 87.3  & 1.745    & N/A           \\
    Tor                                & 100 & Empirical  & \cite{cai-ccs12}        & 65.4      & N/A   & 83.7  & N/A      & N/A           \\
    Tor                                & 120 & Empirical  & this paper              & 56.3      & 36.8  & 77.4  & 1.247    & 4.583 	\footnote{Note that the high latency of TOR is largely due to its onion routing protocols -- a cost that other defenses do not incur.}	 \\
    Tor                                & 200 & Empirical  & this paper              & 50.1      & 31.8  & 75.1  & 1.244    & 4.919	 	    \\
    Tor                                & 775 & Empirical  & \cite{panchenko-wpes11} & 54.6      & N/A   & N/A   & N/A      & N/A           \\
    Tor                                & 800 & Empirical  & \cite{cai-ccs12}        & 40.1      & N/A   & 50.6  & N/A      & N/A           \\
    SSH                                & 120 & Empirical  & this paper              & 86.5      & 75.0  & 80.7  & 1.128    & 1         \\
    SSH                                & 200 & Empirical  & this paper              & 84.4      & 72.9  & 79.4  & 1.111    & 1         \\
    \hline
  \end{tabular}
  \end{center}
  \end{minipage}
  \caption{Main evaluation results for \csb, and comparison to results
    on other schemes reported in other papers.}
  \label{results-summary}
\end{table*}


We have implemented \csb in a custom version of OpenSSH.  Our
implementation also includes a Firefox browser plugin that informs the
SSH client when the browser has finished loading a web page.  The \csb
implementation uses this information to reduce the amount of padding
performed after the page load has completed.

We evaluate \csb, and compare it to Tor, on the Alexa top 200 websites
in the closed-world setting.  The Alexa top 200 websites represent
approximately 91\% of page loads on the internet~\cite{alexa-top-sites}, 
so these results reflect the security users will obtain when using these
schemes in the real world.  Furthermore, prior work on website
fingerprinting attacks has found that an attackers success rate only goes
down as the number of websites increases, so our results give a high-confidence
upper bounds on the success rate these attacks may achieve in larger
settings.

In our experiments, \csb uses 2.8 times as much bandwidth as SSH (i.e. no
defense) and the best known attack had only a 20\% success rate at
inferring which of 200 websites a victim was visiting.  This is a
substantial improvement over previously-proposed schemes -- the same
attack had a success rate over 75\% against Tor and SSH under the same
conditions.

\Cref{results-summary} compares our results with results reported in
other papers.  These comparisons must be done carefully, since the
experiments used different numbers of websites and methodologies.
Nonetheless, the following conclusions are clear from the data:
\begin{itemize}
  \item \csb hides more information than Tor, SSH, HTTPOS, and Tor
    with randomized pipelining, albeit with higher cost. For example,
    the DLSVM attack has a lower success rate against \csb in a
    closed-world experiment with 100 websites than it has against Tor
    with 800 websites.
    
  \item Overall, \csb achieves approximately the same
    bandwidth/security trade-off in our empirical analysis as \buflo
    achieved in Dyer's simulated evaluation.  For example, \csb in
    CTSP mode had a bandwidth ratio of 2.8 and Panchenko's attack had
    a success rate of 23.4\% on 120 websites.  \buflo with $\tau=0$, 
    $\rho=40$, and $d=1500$ had almost identical security, but a
    bandwidth ratio of 2.2.  Although \csb optimizes many aspects of
    the \buflo protocol, an empirical evaluation presents issues that
    do not arise in a simulation, such as dropped packets,
    retransmissions, and application-level timing dependencies.
\end{itemize}

In addition to the empirical work on \csb, this paper provides an
analytical study of the problem of defending against website
fingerprinting attacks.  We show that constructing an
optimally-efficient defense scheme for a given set of websites is an
NP-hard problem.  We then develop lower bounds on the best possible
trade-off between security and overhead that any website
fingerprinting defense can achieve.  Specifically, given a set of
websites and a desired security level, we can compute a lower bound on
the bandwidth overhead that any defense scheme with that security
level can incur on those websites.  This enables us to compare
defenses that offer different security/bandwidth trade-offs by
comparing how close they are to the lower bound.

The paper concludes by using the lower bounds to compare defenses that
offer different bandwidth/security trade-offs.  We find that \csbuflo
gets over 6$\times$ closer to the bandwidth/security trade-off lower
bound than Tor or plain SSH.  Dyer's reported experiments with \buflo
showed somewhat better trade-off performance, but those results were
based on simulations and are not directly comparable.  Despite the
improvement of \csb over Tor and SSH, there is still a large gap
between the lower bounds and the best defenses.

\begin{figure*}[t]
  \centering
  \includegraphics[clip, trim = 0in 5.5in 0in 0in, width=6.25in]{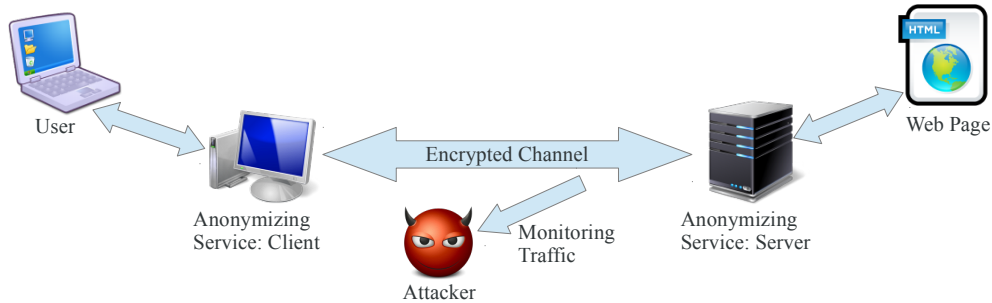}
  \caption{\label{fig:threat-model} Website fingerprinting attack threat model.}
\end{figure*}

In summary, this paper makes the following contributions:
\begin{itemize}
  \item \Cref{sec:theory} provides the first analytical results on
    the website fingerprinting defense problem, showing that
    constructing an optimal defense is NP-hard and discovering lower
    bounds on the best possible trade-off between bandwidth and
    security.
  \item \Cref{sec:design} gives a complete specification of the
    \csb protocol, describing optimizations to make the protocol
    congestion sensitive, rate adaptive, and efficient at hiding
    macroscopic website features, such as total size and the size of
    the last object.
  \item \Cref{subsec:implementation} describes our prototype
    implementation in SSH, which also includes a Firefox plugin to
    notify the proxy when the browser finishes loading a web page.
  \item \Cref{sec:evaluation} presents empirical evaluation results
    for \csb, Tor, and SSH, and shows that \csb provides better
    security, albeit at higher bandwidth costs.  We also show that
    \csb is closer to the lower bound on the security/bandwidth
    trade-off than Tor and SSH.
\end{itemize}


\section{Related Work}
\label{sec:related-work}

\paragraph*{Defenses} Network-level website fingerprinting defenses
pad packets, split packets into multiple packets, or insert dummy
packets.  Dyer, et al., list numerous approaches to padding individual
packets, including pad-to-MTU, pad-to-power-of-two, random padding,
etc.\cite{dyer-snp12}.  They showed that none of the padding schemes
was effective against the attacks they evaluated.  Wright, et al.,
proposed traffic morphing, in which packets are padded and/or
fragmented so that they conform to a specified target
distribution\cite{wright-ndss09}.  Dyer, et al., defeated this
defense, as well\cite{dyer-snp12}.  Lu, et al., extended traffic
morphing to operate on $n$-grams of packet sizes, i.e. their scheme
pads and fragments packets so that $n$-grams of packet sizes match a
target distribution\cite{lu-esorics10}.  Dyer, et al. also proposed
\buflo, which pads or fragments all packets to a fixed size, sends
packets at fixed intervals, injecting dummy packets when necessary,
and always transmits for at least a fixed amount of
time\cite{dyer-snp12}.  They found that they could reduce their best
attack's success rate to 5\% (when guessing from 128 websites), at a
bandwidth overhead of 400\%.  Fu, et al., found in early work
that changes in CPU load can cause slight variations in the time
between packets in schemes that attempt to send packets at fixed
intervals, and recommended randomized inter-packet intervals
instead\cite{fu-iaw03}.

Application-level defenses alter the sequence of HTTP requests and
responses to further obfuscate the user's activity.  For example,
HTTPOS uses HTTP pipelining, HTTP Range requests, dummy requests,
extraneous HTTP headers, multiple TCP connections, and munges TCP
window sizes and maximum segment size (MSS) fields\cite{luo-ndss11}.
Tor has also released an experimental version of Firefox that
randomizes the order in which embedded objects are requested, and the
level of pipelining used by the browser during the
requests\cite{tor-randomized-pipelining}.  Both schemes were defeated
by Cai, et al\cite{cai-ccs12}.

\paragraph*{Attacks}  Researchers have proposed numerous attacks on
basic encrypting tunnels, such as HTTPS, link-level encryption, VPNs,
and IPSec\cite{bissias-pets06, danezis-trafficanalysis, gong-ccs10,
  herrmann-ccsw09, hintz-pets02, liberatore-ccs06, lu-esorics10,
  sun-snp02, yu-tcj11, zhang-wisec11, dyer-snp12}.  These attacks
focus primarily on packet sizes, which carry a lot of information when
no padding scheme is in use.  Herrmann, et al., developed an attack
based on packet sizes that worked well on simple encrypting
tunnels\cite{herrmann-ccsw09}, but performed quite poorly against Tor,
which transmits data in 512-byte cells.  Panchenko, et al., designed
an attack that used packet sizes, along with some ad hoc features
designed to capture higher-level information about the HTTP protocol,
and achieved good success against Tor\cite{panchenko-wpes11}.  Dyer,
et al. performed a comprehensive evaluation of attacks and defenses,
and developed their own attack, called VNG++, that achieved good
success against many network-level defenses\cite{dyer-snp12}.  Cai, et
al., proposed an attack, based on string edit distance, that performs
well against a wide variety of defenses, included application-level
defenses, such as HTTPOS and Tor's randomized
pipelining\cite{cai-ccs12}.  Wang, et al. improved this attack's
performance against Tor by incorporating information about the
structure of the Tor protocol~\cite{wang-wpes13}.
Danezis, Yu, et al., and Cai, et al., all
proposed to use HMMs to extend web page fingerprinting attacks to web
site fingerprinting
attacks\cite{danezis-trafficanalysis,yu-tcj11,cai-ccs12}.

\section{Website Fingerprinting Attacks}
\label{sec:prob-description}

In a website fingerprinting attack, an adversary is able to monitor
the communications between a victim's computer and a private web
browsing proxy, as shown in \Cref{fig:threat-model}.  The private browsing
proxy may be an SSH proxy, VPN server, Tor, or other privacy service.
The traffic between the user and proxy is encrypted, so the attacker
can only see the timing, direction, and size of packets exchanged
between the user and the proxy.  Based on this information, the
attacker attempts to infer the website(s) that the user is visiting
via the proxy.  The attacker can prepare for the attack by collecting
information about websites in advance.  For example, he can visit
websites using the same privacy service as the victim, collecting a
set of website ``fingerprints'', which he later uses to recognize the
victim's site.

Website fingerprinting attacks are an important class of attacks on
private browsing systems.  For example, Tor states that it ``prevents
anyone from learning your location or browsing habits.''\cite{tor-website}  Successful
fingerprinting attacks undermine this security goal.  Fingerprinting
attacks are also a natural fit for governments that monitor their
citizens' web browsing habits.  The government may choose not to (or
be unable to) block the privacy service, but nonetheless wish to infer
citizens' activities when using the service.  Since it can monitor
international network connections, the government is in a good
position to mount website fingerprinting attacks.

Researchers have proposed two scenarios for evaluating website
fingerprinting attacks and defenses: closed-world models and 
open-world models.  A closed-world model consists of a finite number,
$n$, of web pages.  Typical values of $n$ used in past work range from
100 to 800~\cite{dyer-snp12,cai-ccs12,panchenko-wpes11}.  The attacker
can collect traces and train his attack on the websites in the world.
The victim then selects one website uniformly at random, loads it
using some defense mechanism, such as Tor or SSH, and the attacker
attempts to guess which website the victim loaded.  The key
performance metric is the attacker's average success rate.

In an open-world model, there is a population of victims, each of
which may visit any website in the real world, and may select the
website using a probability distribution of their choice. The attacker
does not know any individual victim's distribution over websites, but
has aggregate statistics about website popularity.  The attacker's
goal is to infer which of the victims are visiting a particular
``website of interest'', i.e. an illegal or censored site.  In this
case, the primary evaluation criteria are false positives and false
negatives.  

Perry has critiqued the closed-world model for its
artificiality~\cite{perry-critique}. However, the two models are
connected: Cai, et al., showed how to bootstrap a closed-world attack
into an open-world attack, such that better closed-world performance
yields better open-world performance~\cite{cai-ccs12}.  Thus, although
experiments in the closed-world cannot tell us whether an attack or
defense will be successful in the real world, we can use closed-world
experiments to compare different attacks and defenses.

\section{Theoretical Foundations}
\label{sec:theory}

In this section we focus on understanding the relationship between
bandwidth overhead and security guarantees. We first introduce
definitions of security and overhead for fingerprinting defenses.  We
observe that the overhead required depends on the set of web sites to
be protected -- a set of similar websites can be protected with little
overhead, a set of dissimilar websites requires more overhead.  We
then consider an \emph{offline} version of the website fingerprinting
defense problem, i.e. the defense system knows, in advance, the set of
websites that the user may visit and the packet traces that each
website may generate.  We show that finding a defense system with
optimal overhead in this setting is NP-hard.  We then develop an
efficient dynamic program to compute a lower bound on the bandwidth
overhead of any fingerprinting defense scheme in the closed-world
setting. We will use this algorithm to compute lower bounds on
overhead for the websites used in our evaluation (see
\Cref{sec:evaluation}).

\subsection{Definitions}
\label{subsec:security_def}

In a website fingerprinting attack, the defender selects a website,
$w$, and uses the defense mechanism to load the website, producing a
packet trace, $t$, that is observed by the attacker.  The attacker
then attempts to guess $w$.

Let $W$ be a random variable representing the URL of the website
selected by the defender.  The probability distribution of $W$
reflects the probability that the defender visits each website.  For
each website, $w$, let $T^D_w$ and $T_w$ be the random variables representing the
packet trace generated by loading $w$ with and without defense system $D$, respectively.
Packet traces include the time, direction, and content of each packet.
Since cryptographic attacks are out of scope for this paper, we assume
any encryption functions used by the defense scheme are
information-theoretically secure.  The probability distribution of
$T^D_w$ captures variations in network conditions, changes in
dynamically-generated web pages, randomness in the browser, and
randomness in the defense system.  We assume the attacker knows the
distribution of $W$ and $T^D_w$ for every $w$, so the optimal
attacker, $A$, upon observing trace $t$, always outputs
\[
A(t) = \argmax_w\Pr[W=w]\Pr\left[T^D_w=t\right]
\]
If more than one $w$ attains the maximum, then the attacker chooses
randomly among them.

Some privacy applications require good worst-case performance, and
some only require good average-case performance.  This leads to two
security definitions for website fingerprinting defenses:
\begin{definition}
  Defense $D$ is \textbf{non-uniformly $\epsilon$-secure} if
  $\Pr\left[A(T^D_W)=W\right]\leq\epsilon$.  Defense $D$ is
  \textbf{uniformly $\epsilon$-secure} if
  $\max_w\Pr\left[A(T^D_w)=w\right]\leq\epsilon$.
\end{definition}
These are information-theoretic security definitions -- $A$ is the
optimal attacker described above.  The first definition says that
$A$'s average success rate is less than $\epsilon$, but it does not
require that every website be difficult to recognize.  The second
definition requires all websites to be at least $\epsilon$ difficult
to recognize.  All previous papers on website fingerprinting attacks
and defenses have reported average attack success rates in the
closed-world model, i.e. they have reported non-uniform security
measurements.  We will do the same, although we provide some
comparison with non-uniform security bounds in
\Cref{sec:evaluation}.

To define the bandwidth overhead of a defense system, let $B(t)$ be the
total number of bytes transmitted in trace $t$.  We define the
\textbf{bandwidth ratio} of defense $D$ as
\[
\bwratio_D(W)=\frac{E\left[B\left(T^D_W\right)\right]}{E\left[B\left(T_W\right)\right]}
\]
This definition captures the overall ratio of bandwidth between a user
using defense $D$ for an extended period of time and a user visiting
the same websites with no defense.  

\subsection{Lower Bounds for Bandwidth}\label{subsec:lower_bounds}
In this section we derive an algorithm to compute, given websites
$w_1,\dots,w_n$, a lower bound for the bandwidth that any
deterministic $\epsilon$-secure fingerprinting defense can use in
a closed-world experiment using $w_1,\dots,w_n$.  In a closed-world
experiment, each website occurs with equal probability,
i.e. $\Pr[W=w_i]=\frac{1}{n}$ for all $i$.

To compute a lower bound on bandwidth, we consider an adversary
that looks only at the amount of data transferred by the defense,
i.e. an attacker $A_S$ that always guesses
\[
A_S(t) = \argmax_w\Pr\left[B(T^D_w)=B(t)\right]
\]
Any defense that is $\epsilon$-secure against an arbitrary attacker
must also be at least $\epsilon$-secure against $A_S$.  Thus, if we
can derive a lower bound on defenses that are $\epsilon$-secure
against $A_S$, that lower bound will apply to any $\epsilon$-secure
defense.

We make a few simplifying assumptions in order to obtain an efficient
algorithm for computing lower bounds.  First, we assume that each
website has a unique fixed size, $s_i$.  In our closed world
experiments, we found that, for just over half the web pages in our
dataset, their size had a normalized standard deviation of less than
$0.11$ across $20$ loads, so we do not believe this assumption will
significantly impact the results of our analysis.  Second, we assume
the defense scheme induces a deterministic mapping, $b_i=f(s_i)$, from
the website's original size to the size of the trace observed by the
attacker.  Finally, we assume that the defense mechanism does not
compress or truncate the website, i.e. that $b_i\geq s_i$ for all $i$.

Suppose $f$ is the function induced by such a defense.  Let
$F=\{f(s_1),\dots,f(s_n)\}$.  For any given $b\in F$, let
$n_b=|f^{-1}(b)|$, i.e. the number of websites that cause the defense
mechanism to transmit $b$ bytes.  The probability that the attacker
observes $b$ during a closed world experiment is simply $n_b/n$, and
the probability that the attacker guesses the correct website based on
observation $b$ is $1/n_b$.  Thus the non-uniform security of the
defense scheme is
\[
\sum_{b\in F} \frac{n_b}{n}\frac{1}{n_b}=\frac{|F|}{n}
\]
and the uniform security is $\max_{b\in F} 1/n_b$.  The bandwidth requirements of
the defense is proportional to
\[
\sum_{b\in F}bn_b.
\]
Let $S_b=f^{-1}(b)$.  Since the defense does not compress or truncate
sites, we must have $b\geq \max_{s\in S_b}s$.  For the purposes of computing
lower bounds on the bandwidth, we may as well assume that $b=\max_{s\in S_b} s$.
Thus the function $f$ is equivalent to a partition of the set
$\{s_1,\dots,s_n\}$.

These observations imply that the optimal $f$ must be monotonic.
\begin{theorem}\label{th:contiguous}
  The optimal $f$ is monotonic.
\end{theorem}
\begin{proof}
Consider any partition of $\{s_1,\dots,s_n\}$ into sets $S_1, \dots,
S_k$.  Let $m_i=\max_{s\in S_i} S_i$. Without loss of generality,
assume $m_1 \leq m_2 \leq \dots \leq m_k$. Now consider the monotonic
allocation of traces into sets $S^*_1, \dots, S^*_{k}$ where $|S^*_i|
= |S_i|$. Let $m^*_i=\max_{s\in S^*_i}s$. Observe that $m^*_i \leq
m_i$ for all $i$, i.e. the new allocation has lower bandwidth.

Since the number of sets in the partition and the sizes of those sets
are unchanged, this new allocation has the same uniform and
non-uniform security as the original, but lower bandwidth.  Hence the
optimal $f$ must be monotonic.
\end{proof}

We can compute the optimal partition for a given security parameter
using a dynamic program.  If $S_1, \dots, S_k$ is is an optimal
uniformly $\epsilon$-secure partition, then so is $S_1,\dots,S_{k-1}$.
Thus the cost, $C(\epsilon, n)$ of the optimal uniformly
$\epsilon$-secure partition satisfies the recurrence relation:
\[
C(\epsilon, n)=
\left\{
\begin{array}{ll}
  \infty & \text{if $n<1/\epsilon$} \\
  \displaystyle\min_{1\leq j\leq n-1/\epsilon}C(\epsilon, j)+(n-j)s_n & \text{otherwise.}
\end{array}\right.
\]
Non-uniformly $\epsilon$-secure partitions satisfy a slightly
different recurrence.  If $S_1, \dots, S_k$ is is an optimal
non-uniformly $\frac{k}{n}$-secure partition, then $S_1,\dots,S_{k-1}$
is an optimal non-uniformly $\frac{k-1}{n-|S_k|}$-secure partition.
Therefore the optimal cost, $C'(\frac{k}{n},n)$, satisfies the
recurrence
\[
C'(\frac{k}{n}, n)=
\left\{
\begin{array}{ll}
  ns_n & \text{if $k=1$} \\
  \displaystyle\min_{1\leq j\leq n-1}C'(\frac{k-1}{j}, j)+(n-j)s_n & \text{o.w.}
\end{array}\right.
\]

\Cref{alg:dp} shows a dynamic program for computing a lower bound on
the bandwidth of any deterministic defense that can achieve $\epsilon$
non-uniform security in a closed-world experiment on static websites
with sizes $s_1,\dots,s_n$.  We use this algorithm to compute the
lower bounds reported in \Cref{sec:evaluation}.

\begin{algorithm}[t]
  \caption{Algorithm to compute a lower bound on the bandwidth of any offline non-uniformly $\epsilon$ secure fingerprinting defense against $A^{S}$ attackers.}
  \label{alg:dp}
  \begin{algorithmic}
    \Function{$A^{S}$-min-cost}{$n$, $\epsilon$, $\{s_1, \dots, s_n\}$} 
    \State Array $C[0 \dots n\epsilon,0 \dots n]$
    \For{$i=0,\dots,n\epsilon$}
	    \State $C[i, 0] \gets 0$
    \EndFor
    \For{$i=0,\dots,n$}
      \State $C[0, i] \gets \infty$
    \EndFor
    \For{$i = 1 \to n$}
      \For{$j = 1 \to n\epsilon$}
        \State $C[j, i] = \min_{1\leq \ell \leq i-1} \left[(i-\ell) s_{i} + C[j-1, \ell]\right]$
	    \EndFor
	  \EndFor
	  \State \Return $C[n\epsilon, n]$
    \EndFunction
  \end{algorithmic}
\end{algorithm}

\subsection{Security Against DLSVM Attackers}

We now analyze the task of defending against DLSVM-style
attackers in the same theoretical setting as above.  We will show that
finding the lowest-cost offline defense against a DLSVM attacker
is NP-hard, via a reduction from the binary shortest common
super-sequence problem.  This reduction will also show that the
minimum bandwidth required by an offline defense against a DLSVM
attacker is at most twice the bandwidth lower bound computed in the
previous section.  This result, along with the experimental results in
\Cref{sec:evaluation}, will show that offline defenses can achieve low
cost and high security, suggesting a promising avenue for future
work.

Suppose websites $w_1,\dots,w_n$ are all static and constructed such
that loading each site requires performing a fixed, serialized
sequence of requests and responses, e.g. each web page contains a
javascript program that loads objects one at a time in a fixed order.
Let $d_i[j]=1$ iff the $j$th byte that must be transmitted to load
page $w_i$ is a transmission in the upstream direction.  

Loading website $w_i$ via a deterministic defense mechanism produces a
fixed trace $t_i$. Let $z_i$ be the binary string defined by
$z_i[j]=1$ iff the $j$th byte of $t_i$ is an upstream byte.  Since,
for these websites, the defense mechanisms cannot delete or re-order
bytes, we must have that $d_i$ is a sub-sequence of $z_i$.

When the victim loads a web site, producing trace $t$, the attacker
can compute the corresponding string, $z$.  In order for the attacker
to learn nothing about which web page the victim loaded, we must have
that, for all $i$, $d_i$ is a substring of $z$.
Thus
the defense system must compute some string, $z$, that is
simultaneously a super-sequence of $d_1,\dots,d_n$.  Minimizing the
cost of such a defense is thus equivalent to finding the
shortest common super-sequence (SCS) of $d_1, \dots, d_n$.
This problem is NP-hard\cite{SCS_SR8}.

However, there is a simple 2-approximation for the binary SCS problem.
Let $\ell$ be the length of the longest string $d_1,\dots,d_n$.  Their
SCS must be at least $\ell$ long, but is at most $2\ell$ long, since
every binary string of length at most $\ell$ is a sub-sequence of
$(01)^\ell$.  Thus for any set of static websites $w_1,\dots,w_n$,
there exists a deterministic offline defense that achieves (uniform or
non-uniform) $\epsilon$-security against DLSVM-style attackers and
incurs bandwidth cost that is at most twice the bandwidth lower bound
derived in the previous section.


\section{\csbuflo}
\label{sec:design}

Dyer, et al., described \buflo, a hypothetical defense scheme that
hides all information about a website, except possibly its size, and
performed a simulation-based evaluation that found that, although
\buflo is able to offer good security, it incurs a high cost to do so.

In this section, we describe \csbuflo (\csb), an extension to \buflo that
includes numerous security and efficiency improvements.  \csb
represents a new approach to the design of fingerprinting defenses.
Most previously-proposed defenses were designed in response to known
attacks, and therefore took a black-listing approach to information
leaks, i.e. they tried to hide specific features, such as packet
sizes.  In designing \csb, we take a white-listing approach -- we
start with a design that hides all traffic features, and iteratively
refine the design to reveal certain traffic features that enable us to
achieve significant performance improvements without significantly
harming security.

\subsection{Review of \buflo}

The Buffered Fixed-Length Obfuscator (\buflo) of Dyer, et al.,
transmits a packet of size $d$ bytes every $\rho$ milliseconds, and
continues doing so for at least $\tau$ milliseconds.  If $b<d$ bytes
of application data are available when a packet is to be sent, then
the packet is padded with $d-b$ extra bytes of junk.  The protocol
assumes that the junk bytes are marked so that the receiver can
discard them.  If the website does not finish loading within $\tau$
milliseconds, then \buflo continues transmitting until the website
finishes loading and then stops immediately.  Dyer, et al., did not
specify how \buflo detects when the website has finished loading.
They also did not specify how \buflo handles bidirectional
communication -- presumably independent \buflo instances are run at
each end-point.

\buflo effectively hides everything about the website, except possibly
its size, but has several shortcomings:
\begin{itemize}
  \item It either completely hides the size of the website or
    completely reveals it ($\pm d$ bytes).  Thus it does not provide
    the same level of security to all websites.
  \item \buflo has large overheads for small websites.  Thus its
    overhead is also unevenly distributed.
  \item \buflo is not TCP-friendly.  In fact, it is the epitome of a
    bad network citizen.
  \item \buflo does not adapt when the user is visiting fast or slow
    websites.  It wastes bandwidth when loading slow sites,
    and causes large latency when loading fast websites.
  \item \buflo must be tuned to each user's network connection.  If
    the \buflo bandwidth, $\frac{1000d}{\rho}$ B/s, exceeds the user's
    connection speed, then \buflo will incur additional delay without
    improving security.
  \item Past research by Fu, et al., showed that transmitting at fixed
    intervals can reveal load information at the sender, which an
    attacker can use to infer partial information about the data being
    transmitted\cite{fu-iaw03}.
\end{itemize}
Dyer, et al., proposed \buflo as a straw-man defense system, so it is
understandable that they did not bother addressing these problems.
However, we show below that several of these problems have common
solutions, e.g. we can simultaneously improve overhead and
TCP-friendliness, simultaneously make security and overhead more
uniform, etc.  Thus, as our evaluation will show, \csb may be a
practical and efficient defense for users requiring a high level of
security.

Further, as noted by its authors, \buflo's simulation based results
``reflect an ideal implementation that assumes the feasibility of
implementing fixed packet timing intervals. This is at the very
least difficult and clearly impossible for certain values of $\rho$.
Simulation also ignores the complexities of cross-layer communication 
in the network stack'' \cite{dyer-snp12}. As a result, it remains 
unclear how well the defense performs in the real world.

\subsection{Overview of \csbuflo}

\Cref{alg:csbuflo-server} shows the main loop of the \csb server.  The
client loop is similar, except for the few differences discussed
throughout this section.  Similar to \buflo, \csb delivers
fixed-size chunks of data at semi-regular intervals.  \csb randomizes
the timing of network writes in order to counter the attack of Fu, et
al.\cite{fu-iaw03}, but it maintains a target average inter-packet
time, $\rho^*$.  \csb periodically updates $\rho^*$ to match its
bandwidth to the rate of the sender (\Cref{ssec:rate-adaptation}).
Since updating $\rho^*$ based on the sender's rate reveals information
about the sender, \csb performs these updates infrequently.  \csb uses
TCP to be congestion friendly, and uses feedback from the TCP stack in
order to reduce the amount of junk data it needs to send
(\Cref{ssec:congestion-sensitivity}).  Also like \buflo, \csb
transmits extra junk data after the website has finished loading in
order to hide the total size of the website.  However, \csb uses a
scale-independent padding scheme (\Cref{ssec:stream-padding}) and
monitors the state of the page loading process to avoid some
unnecessary overheads  (\Cref{ssec:early-termination}).

\begin{algorithm}[ht!]
  \caption{The main loop of the \csbuflo server.}
  \label{alg:csbuflo-server}
  \begin{algorithmic}
    \Function{CSBUFLO-Server}{$s$}
      \While{true}
        \State $(m,\rho) = \Call{read-message}{\rho}$
        \If{$m$ is application data from website}
			\State \emph{output-buff} $\gets$ \emph{output-buff} $\|$ \emph{data}
          \State \emph{real-bytes} $\gets$ \emph{real-bytes} + \Call{length}{m}
          \State \emph{last-site-response-time} $\gets$ \Call{current-time}{}
        \ElsIf{$m$ is application data from client}
          \State send $m$ to the website
		  \State \emph{$\rho$-stats} $\gets$ \emph{$\rho$-stats} $\| \perp$ 
          \State \emph{onLoadEvent} $\gets$ 0, \emph{padding-done} $\gets$ 0
        \ElsIf{$m$ is onLoad message}
          \State \emph{onLoadEvent} $\gets$ 1
        \ElsIf{$m$ is padding-done message}
          \State \emph{padding-done} $\gets$ 1
        \ElsIf{$m$ is a time-out}
          \If{\emph{output-buff} is not empty}
		  \State \emph{$\rho$-stats} $\gets$ \emph{$\rho$-stats} $\|$ \Call{current-time}{} 
          \EndIf
          \State (\emph{output-buff}, $j)$ $\gets$ \Call{cs-send}{s, \emph{output-buff}}
          \State \emph{junk-bytes} $\gets$ \emph{junk-bytes} + \emph{j}
        \EndIf
        \If{\Call{done-xmitting}{}}
			\State reset all variables
        \Else\Comment{$\rho^*:$ Average time between sends to client}
          \If{$\rho^*=\infty$} 
            \State $\rho^* \gets \initialrho$
          \ElsIf{\Call{crossed-threshold}{\emph{real-bytes, junk-bytes}}}
            \State $\rho^* \gets \Call{rho-estimator}{\rho$\emph{-stats},$\rho^*}$
            \State $\rho$\emph{-stats} $\gets$ $\emptyset$ 
          \EndIf
          \\
          \If{$m$ is a time-out}
            \State $\rho \gets$ random number in $[0,2\rho^*]$
          \EndIf
        \EndIf
      \EndWhile
    \EndFunction
  \end{algorithmic}
\end{algorithm}

\subsection{Rate Adaptation}
\label{ssec:rate-adaptation}


\begin{figure*}[t]
  \centering
  \includegraphics[width=6in]{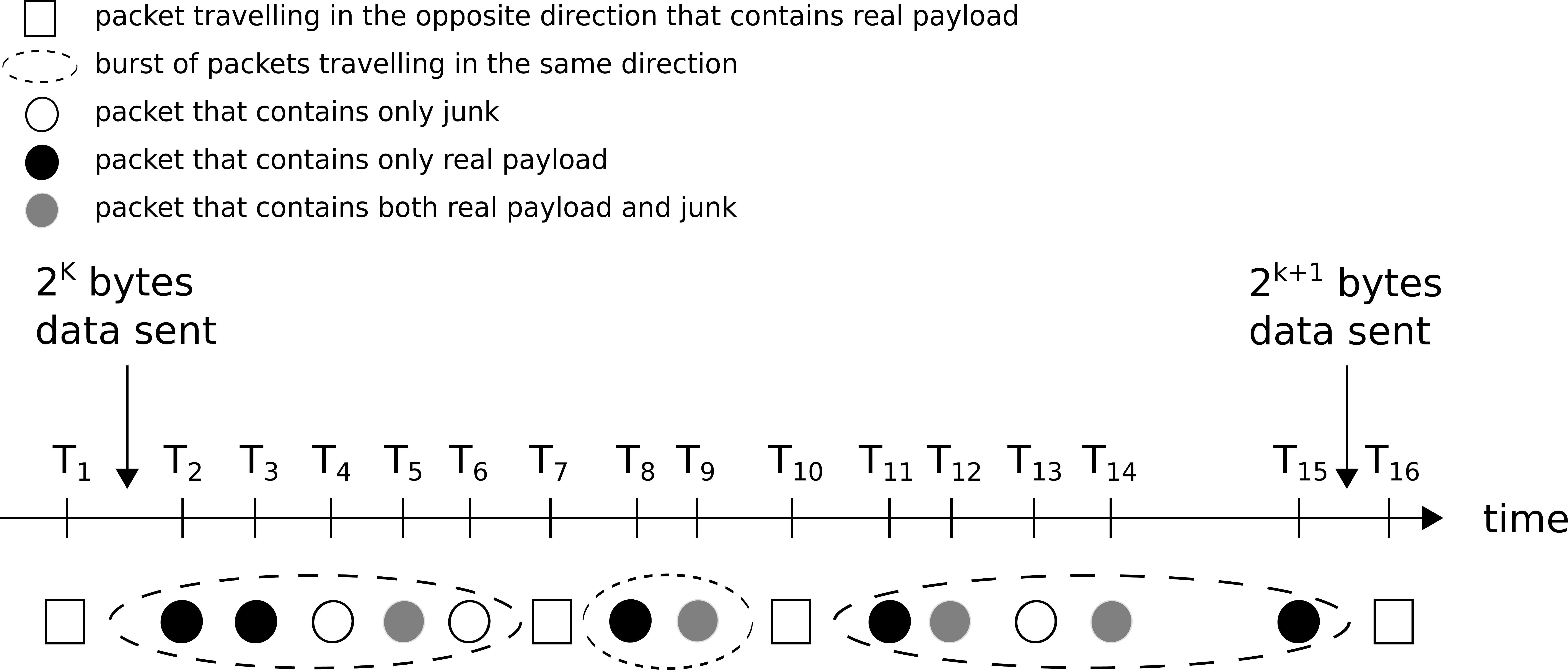}
  \caption{Rate adaptation in \csb. $\rho^*$ is updated based on the packets
	  transmitted to the other end between $T_2$ and $T_{15}$. Time intervals
	  between two consecutive packets are stored in an array $Intervals[]$.
	  The two packets under consideration both contain some real payload
	  data and they belong to the same burst. i.e. $Intervals = [T_3-T_2,
	  T_5-T_3, T_9-T_8, T_{12}-T_{11}, T_{14}-T_{12}, T_{15}-T_{14}]$
	  and $\rho^*=2^{\lfloor\log_2 Median(Intervals[])\rfloor}$.}
  \label{fig:rate-adaptation}
\end{figure*}

\begin{algorithm}[t]
  \caption{Algorithm for estimating new value of $\rho^*$ based on past network performance.}
  \label{alg:rho-estimator}
  \begin{algorithmic}
    \Function{rho-estimator}{$\rho$\emph{-stats}, $\rho^*$} 
    \State $I \gets \left[\rho\emph{-stats}_{i+1}-\rho\emph{-stats}_{i}\mid \rho\emph{-stats}_{i}\not=\perp\wedge \rho\emph{-stats}_{i+1}\not=\perp\right]$
    \If{$I$ is empty list}
      \State\Return $\rho^*$
    \Else
      \State\Return$2^{\lfloor\log_2\Call{median}{I}\rfloor}$
    \EndIf
    \EndFunction
  \end{algorithmic}
\end{algorithm}

\csb adapts its transmission rate to match the rate of the sender.
This reduces wasted bandwidth when proxying slow senders, and it
reduces latency when proxying fast senders.  However, adapting
\csb's transmission rate to match the sender's reveals information
about the sender, and therefore may harm security.

As shown in Figure~\ref{fig:rate-adaptation}, \csb takes several steps to limit
the information that is leaked through rate adaptation.  First, it only adapts
after transmitting $2^k$ bytes, for some integer $k$.  Thus, during a session
in which \csb transmits $n$ bytes, \csb will perform $\log_2 n$ rate
adjustments, limiting the information leaked from these adjustments.  This
choice also allows \csb to adapt more quickly during the beginning of a
session, when the sender is likely to be performing a TCP slow start.  During
this phase, \csb is able to ramp up its transmission rate just as quickly as
the sender can.

\csb further limits
information leakage by using a robust statistic to update $\rho^*$.
Between adjustments, it collects estimates of the sender's
instantaneous bandwidth. It then sets $\rho^*$ so as to match the
sender's median instantaneous bandwidth.  Median is a robust
statistic, meaning that the new $\rho^*$ value will not be strongly
influenced by bandwidth bursts and lulls, and hence $\rho^*$ will not
reveal much about the sender's transmission pattern.

Note that the estimator only collects measurements during
uninterrupted bursts from the sender.  This ensures that the bandwidth
measurements do not include delays caused by dependencies between
requests and responses.  

For example, if the estimator sees a packet
$p_1$ from the website, then a packet $p_2$ from the client, and then
another packet $p_3$ from the website, it may be the case that $p_3$ is
a response to $p_2$.  In this case, the time between $p_1$ and $p_3$
is constrained by the round trip time, not the website's bandwidth.

Finally, \csb rounds all $\rho^*$ values up to a power of two.
This further hides information about the sender's true rate, and gives
the sender room to increase it's transmission rate, e.g. during slow
start.

\subsection{Congestion-Sensitivity}
\label{ssec:congestion-sensitivity}

There's a trivial way to make \buflo congestion sensitive and TCP
friendly: run the protocol over TCP. With this approach, we grab an 
additional opportunity for increasing efficiency: when the network is
congested, \csb does not need to insert junk data to fill the output
buffer.

\begin{algorithm}[t]
  \caption{Algorithm for sending data and using feedback from TCP.
    Socket $s$ should be configured with \texttt{O\_NONBLOCK}.}
  \label{alg:cssend}
  \begin{algorithmic}
    \Function{cs-send}{$s$, \emph{output-buff}}
      \State $n \gets$\Call{length}{\emph{output-buff}}
      \State $j\gets 0$
      \If{$n<\packetsize$}
        \State $j \gets \packetsize-n$
        \State  \emph{output-buff} $\gets$  \emph{output-buff} $\|$ \emph{j}
      \EndIf
      \State $r \gets$  \texttt{write}($s$, \emph{output-buff}, \packetsize)
      \If{$r\geq n$} \Comment{Optional: reclaim unsent junk}
        \State $\emph{output-buff} \gets $ empty buffer
        \State $j \gets r - n$
      \Else
        \State remove last $j$ bytes from $\emph{output-buff}$
        \State remove first $r$ bytes from $\emph{output-buff}$
        \State $j \gets 0$
      \EndIf
      \State \Return $(\emph{output-buff},j)$
    \EndFunction
  \end{algorithmic}
\end{algorithm}

\Cref{alg:cssend} shows our method for taking advantage of congestion
to reduce the amount of junk data sent by \csb.  Note first that
\textsc{cs-send} always writes exactly $d$ bytes to the TCP socket.
Since the amount of data presented to the TCP socket is always the
same, this algorithm reveals no information about the timing or size
of application-data packets from the website that have arrived at the
\csb proxy.

This algorithm takes advantage of congestion to reduce the amount of
junk data it sends.  To see why, imagine the TCP connection to the
client stalls for an extended period of time.  Eventually, the
kernel's TCP send queue for socket $s$ will fill up, and the call to
\texttt{write} will return $0$.  From then until the TCP congestion
clears up, \csb calls to \textsc{cs-send} will not append any further
junk data to $B$.  

\subsection{Stream Padding}
\label{ssec:stream-padding}

\csb hides the total size of real data transmitted by continuing to
transmit extra junk data after the browser and web server have stopped
transmitting.

\begin{table}[t]
  \small
  \centering
  \bgroup
  \def\arraystretch{1.5}
  \begin{tabular}{|p{1.1cm}|p{2cm}|p{2cm}|p{2.2cm}|}
    \hline
    Padding Schemes      		&	Payload Sent \newline Before Padding	& Junk Sent	\newline Before Padding		&	Total Bytes Sent\newline After Padding \\
    \hline
    \hline
    \textit{payload} padding	&	$R$				&	$J$			&	$c2^{\lceil\log_2 R\rceil}$		\\
	\hline    
    \textit{total} padding		&	$R$				&	$J$			& 	$2^{\lceil\log_2 (R+J)\rceil}$		\\
    \hline
  \end{tabular}
  \egroup
  \caption{Two different padding schemes for \csb.}
  \label{tab:padding-schemes}
\end{table}

\Cref{tab:padding-schemes} shows two related padding schemes we experimented with in \csb.  Both
schemes introduce at most a constant factor of additional cost,
but reveal at most a logarithmic amount of information about the size
of the website.  The first scheme, which we call \textit{payload}
padding, continues transmitting until the total amount of transmitted
data ($R+J$) is a multiple of $2^{\lceil\log_2 R\rceil}$.  This
padding scheme will transmit at most $2^{\lceil\log_2 R\rceil}$
additional bytes, so it increases the cost by at most a factor of
$2$, but it reveals only $\log_2 R$.

The second scheme, which we call \textit{total} padding, continues
transmitting until $R+J$ is a power of $2$.  This also increases the
cost by at most a factor of $2$ and reveals, in the worst case,
$\log_2 R$, but it will in practice hide more information about $R$
than payload padding.

Note that the \csb server and the \csb client do not have to use the
same stream padding scheme.  Thus, there are four possible padding
configurations, which we denote CPSP (client payload, server payload),
CPST (client payload, server total), CTSP (client total, server
payload) and CTST (client total, server total).

In order to determine when to stop padding, the \csb server must know
when the website has finished transmitting.  \csbuflo uses two
mechanisms to recognize that the page has finished loading.  First,
the \csb client proxy monitors for the browser's onLoad event.  The
\csb client notifies the \csb server when it receives the onLoad event
from the browser.  Once the \csb server receives the onLoad message
from the client, it considers the web server to be idle (see
\Cref{alg:csbuflo-helpers}) and will stop transmitting as soon as it
adds sufficient stream padding and empties its transmit buffer.  As a
backup mechanism, the \csb server considers the website idle if
\quiettime seconds pass without receiving new data from the website.
We used a \quiettime of 2 seconds in our prototype implementation.


\begin{algorithm}[t]
  \caption{Definition of the \textsc{done-xmitting} function.}
  \label{alg:csbuflo-helpers}
  \begin{algorithmic}
    \Function{done-xmitting}{}
      \State \Return $\Call{length}{\emph{output-buff}} \gets 0$  $\wedge\Call{channel-idle}{\emph{onLoadEvent,last-site-response-time}} \wedge $ $(\emph{padding-done} $ $\vee 
      $ $\Call{crossed-threshold}{\emph{real-bytes + junk-bytes}})$
    \EndFunction
    \\
    \Function{channel-idle}{$\emph{onLoadEvent}$, $\emph{last-site-response-time}$}
      \State\Return $\emph{onLoadEvent} \vee (\emph{last-site-response-time} + \quiettime<\Call{current-time}{})$
    \EndFunction
    \\
    \Function{crossed-threshold}{$x$}
      \State\Return $\lfloor\log_2 (x-\packetsize)\rfloor<\lfloor\log_2 x\rfloor$
    \EndFunction

  \end{algorithmic}
\end{algorithm}

\subsection{Early Termination}
\label{ssec:early-termination}

As described above, the \csb server is likely to finish each page load
by sending a relatively long tail of pure junk packets.  This tail can
be a significant source of overhead and, somewhat surprisingly, may
not provide much additional security.  

Our initial investigations revealed that the long tail served two
purposes which could also be served through other, more efficient
means.  As mentioned above, the long tail helps hide the total size of
the website.  However, the interior padding performed by
\textsc{cs-send} also obscures the total size of the website.  Our
evaluation in \Cref{sec:evaluation} investigates the security impact of
additional stream padding.

In the specific context of web browsing, the long tail also hides the
size of the last object sent from the web server to the client.  The
attacker can infer some information about the size of this object by
measuring the amount of data the \csb server sends to the \csb client
after the \csb client stops transmitting to the \csb server.  However,
this information can also be hidden by having the \csb client continue
to send junk packets to the \csb server, i.e. more aggressive stream
padding from the \csb client may obviate the need for aggressive
padding at the \csb server.  

Based on these ideas, we implemented an \textit{early termination}
feature in our \csb prototype.  The \csb client notifies the \csb
server that it is done padding.  After receiving this message, the
\csb server will stop transmitting as soon as the web server becomes
idle and its buffers are empty.  

\begin{figure}[t]
  \scalebox{.525}{\includegraphics{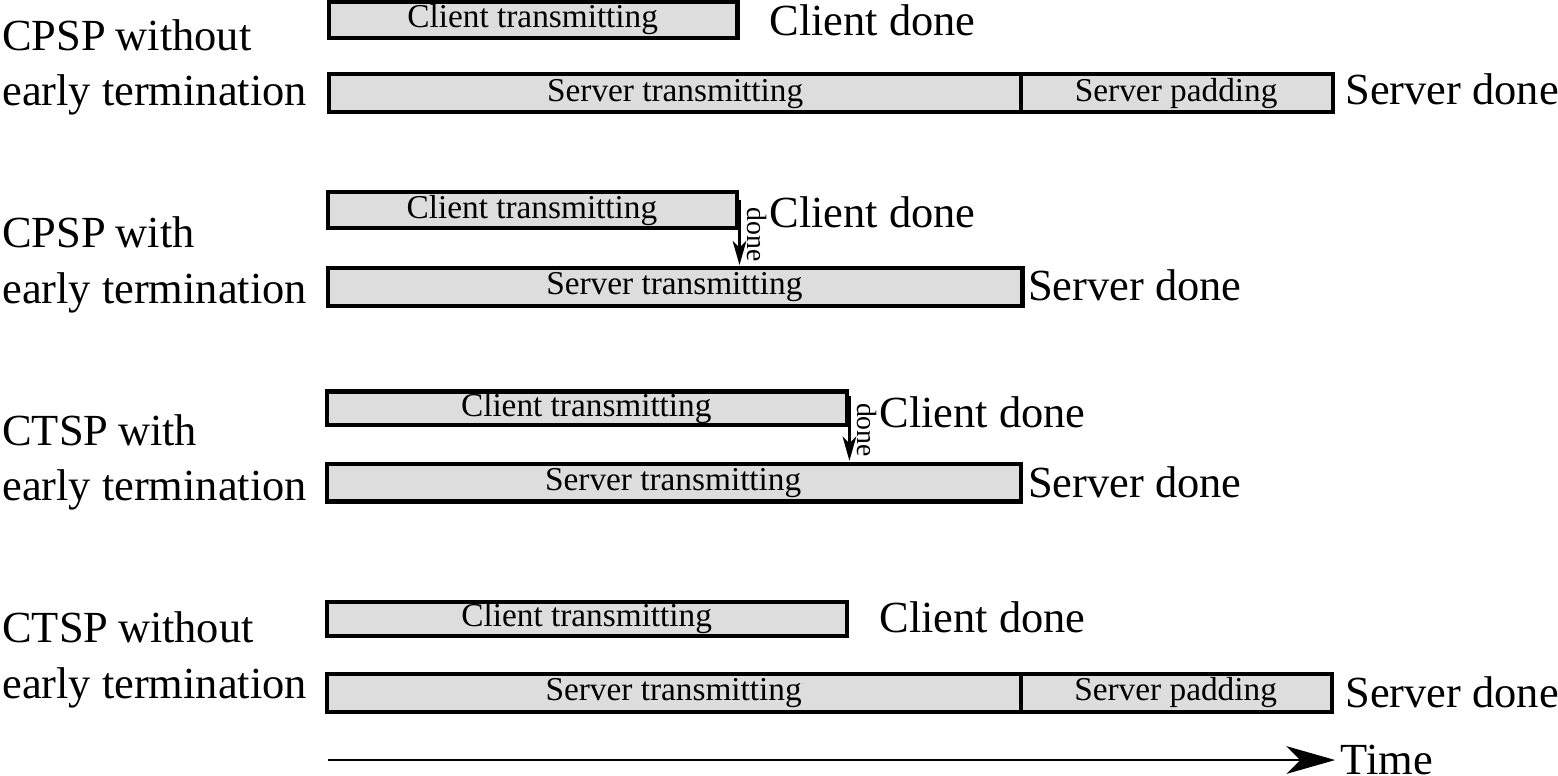}}
  \caption{The interaction between client and server padding schemes
    and early termination.  More padding at the client can help hide
    the size of the last object sent from the server to the client.
    Early termination can avoid unnecessary padding at the end of a
    page load.}
  \label{fig:padding-schemes}
\end{figure}

\Cref{fig:padding-schemes} illustrates how the padding scheme used by
the client and server can interact, including the impact of early
termination.  Additional client padding can hide the size of the last
HTTP object, and early termination can avoid unnecessary padding.  Our
evaluation investigates the security/efficiency trade-offs between
different padding regimes at the client and server, and how they
interact with early termination.

\subsection{Packet Sizes}
\label{subsubsec:hiding-packet-sizes}
Sending fixed-length packets hides packet size information from the
attacker.  Although any fixed length should work, it is important to
choose a packet length that maximizes performance. Since we may
transmit pure junk packets during the transmission, larger packets
tend to cause higher bandwidth overhead, and on the other hand,
smaller packets may not make full use of the link between the client
and server, thus increase the loading time. 

Preliminary investigations revealed that over 95.7\% of all upstream
packet transmissions are under 600 bytes, therefore, this was used
as the standard packet size in our experiments.

\section{Prototype Implementation}
\label{subsec:implementation}

We modified OpenSSH5.9p1 to implement \Cref{alg:csbuflo-server}.
However, the optional junk 
recovery algorithm described in \Cref{alg:cssend} was not implemented.

The SSH client was also modified to accept a new SOCKS proxy command 
code, \textit{onLoadCmd}. This command was used to communicate to the 
server when to stop padding (as described in \Cref{ssec:stream-padding}).
A Firefox plugin, \textit{OnloadNotify}, that, upon detecting the page 
\textit{onLoad} event, connects to the SSH client's SOCKS port and issues 
the \textit{onLoadCmd}, was also developed. 

In addition, the following OpenSSH message types were used:
\begin{enumerate}
\item The OpenSSH message type \texttt{SSH\_MSG\_IGNORE}, which means 
all payload in a packet of this type can be ignored, was used 
to insert junk data whenever needed.    
\item The \texttt{SSH\_MSG\_NOTIFY\_ONLOAD} message was created to be
used by the client to communicate reception of \textit{onLoadCmd} from 
the browser, to the server. Upon receiving this message from the client,
the \csb server stops transmitting as soon as it empties its buffer and
adds sufficient stream padding.
\item The \texttt{SSH\_MSG\_NOTIFY\_PADDINGDONE} message was created to 
implement the early termination feature of \csb. Upon receiving this message
from the client, the \csb server stops transmitting as soon as the web 
server becomes idle and its buffers are empty.
\end{enumerate}

All the above messages were buffered and transmitted just like other 
messages in \Cref{alg:csbuflo-server}, i.e. using \textsc{cs-send}, 
therefore an attacker is unable distinguish these messages from 
other traffic.

\section{Evaluation}
\label{sec:evaluation}

We investigated several questions during our evaluation:
\begin{itemize}
  \item How do the different stream padding schemes affect performance
    and security of \csb?  What is the effect of adding early
    termination to the protocol?
  \item How does \csb's security and overhead compare to Tor's, and
    how do they both compare to the theoretical minimums derived in
    \Cref{sec:theory}?  
  \item Can we use the theoretical lower bounds to enable us to
    compare defenses that have different security/overhead trade-offs?
\end{itemize}

\subsection{Experimental Setup}
\label{subsec:setup}

For our main experiments, we collected traffic from the Alexa top 200
functioning, non-redirecting web pages using four different defenses:
plain SSH, Tor, \csb with the CTSP padding and early termination, and
\csb with CPSP padding and early termination.  We also collected
several smaller data sets using other configurations of \csb, but
these are only used in the padding scheme evaluations
(\Cref{tab:padding-scheme-results}).

We constructed a list of the Alexa top 200 functioning,
non-redirecting, unique pages, as follows.  We removed web pages that
failed to load in Firefox (without Tor or any other proxy).  We
replaced URLs that redirected the browser to another URL with their
redirect target.  Some websites display different languages and
contents depending on where the page is loaded,
e.g. \textit{www.google.com} and \textit{www.google.de}. We kept only
one URL for this type of website, i.e. we only had
\textit{www.google.com} in our set.  Our data set consisted of Alexa's
200 highest-ranked pages that met these criteria.

We collected 20 traces of each URL, clearing the browser cache
between each page load. We collected traces from each web page in a
round-robin fashion.  As a result, each load of the same URL occurred
about 5 hours apart.

Measuring the precise latency of a fingerprinting defense scheme poses a
challenge: we can easily measure the time it takes to load a page
using the defense, but we cannot infer the exact time it \emph{would have
taken} to load the page without the defense.  Therefore, every time we 
loaded a page using a defense, we immediately loaded it again using SSH 
to get an estimate of the time it would have taken to load the page without
the defense in place.  We then compute latency ratios the same way
we compute bandwidth ratios, i.e. if $L(t)$ is the total duration
of a packet trace, the latency ratio of a defense scheme is
\[
\frac{E\left[L(T^D_W)\right]}{E\left[L(T_W)\right]}
\]

We collected network traffic using several different computers with
slightly different versions of Ubuntu Linux -- ranging from 9.10 to
11.10.  We used Firefox 3.6.23-3.6.24 and Tor 0.2.1.30 with polipo
HTTP Proxy. All Firefox plugins were disabled during data collection,
except when collecting \csb traffic, where we enabled the
\textit{OnloadNotify} plugin.  Three of the computers had 2.8GHz Intel
Pentium CPUs and 2GB of RAM, one computer had a 2.4GHz Intel Core 2
Duo CPU with 2GB of RAM.  We scripted Firefox using Ruby and captured
packets using tshark, the command-line version of wireshark.  For the
SSH experiments, we used OpenSSH5.3p1. Our Tor clients used the
default configuration. SSH tunnels passed between two machines on the
same local network.

We measured the security of each defense by using the three best
traffic analysis attacks in the literature: VNG++ \cite{dyer-snp12},
the Panchenko SVM \cite{panchenko-wpes11}, and DLSVM
\cite{cai-ccs12}. We ran each of the above classifiers against 
the traces generated by each defense using stratified 10-fold cross 
validation.

\subsection{Results}
\label{subsec:results}



\paragraph*{Padding Schemes} \Cref{tab:padding-scheme-results} shows 
the bandwidth ratio, latency ratio, and security (estimated
using the VNG++ attack) of four different versions of \csb on a data
set of 50 websites.  Note that early termination does not appear to
affect security, although it can significantly reduce overhead in some
configurations.  All other experiments in this paper use early
termination.  The client padding scheme, on the other hand, appears to
control a trade-off between security and overhead.  Therefore we
report the rest of our results for both CPSP and CTSP padding.

\begin{table}[t]
  \small
  \centering
  \bgroup
  \def\arraystretch{1.5}
  \begin{tabular}{|p{1.1cm}p{1.6cm}p{1.5cm}p{1.05cm}p{1.25cm}|}
    \hline
    Padding       &	Early\newline Termination	&	Bandwidth Ratio & Latency Ratio & VNG++\newline Accuracy	\\
    \hline
    \hline
    CTSP	&	Yes               &	3.59               & 3.91    & 29.0\%               \\
    CTSP	&	No                &	3.73               & 3.51    & 29.6\%             \\
    CPSP	&	Yes               &	2.60               & 2.87    & 34.2\%             \\
    CPSP	&	No                &	3.42               & 3.52    & 36.0\%               \\
    \hline
  \end{tabular}
  \egroup
  \caption{Security and performance of \csbuflo variants. VNG++ success rate is the
    probability that the attack was able to correctly guess which of
    50 web pages the user was visiting.}
  \label{tab:padding-scheme-results}
\end{table}

\paragraph*{Security Comparison}
Figure~\ref{fig:all-attacks-vs-sites-various} shows the level of
security various defense schemes provide against three different
attacks, as the number of web pages the attacker needs to distinguish
increases. Note that the \csb schemes have significantly better
security than Tor and SSH.  For each defense scheme, we compute its
average bandwidth ratio, $BO$, and plot the lower bound on security
that can be achieved within that ratio, using the algorithm from
\Cref{sec:theory}.

\begin{figure*}[t]
  \centering
  \subfigure[\csb (CTSP)]{
    \includegraphics[clip, trim = 1.1in 6.9in 3.5in 1.75in, width=3.35in]{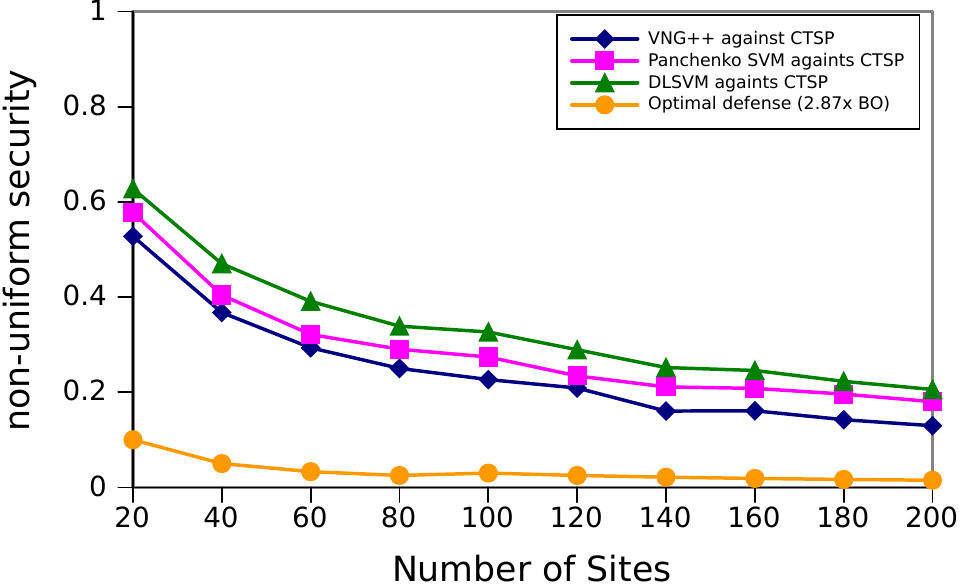}
    \label{fig:all-attacks-vs-sites-CTSP}
  }
  \subfigure[\csb (CPSP)]{
    \includegraphics[clip, trim = 1.1in 6.9in 3.5in 1.75in, width=3.35in]{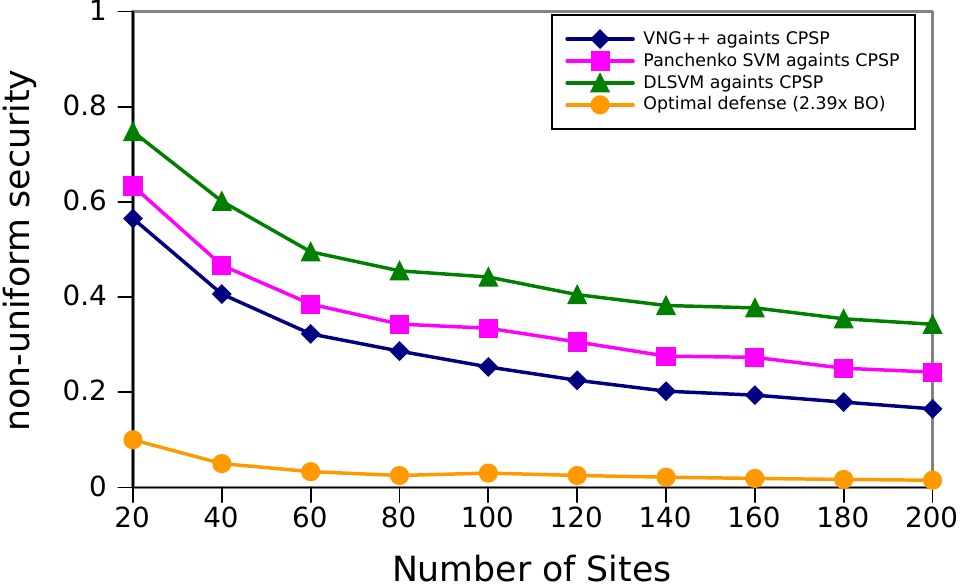}
    \label{fig:all-attacks-vs-sites-CPSP}
  }
  \subfigure[Tor]{
    \includegraphics[clip, trim = 1.1in 6.9in 3.5in 1.75in, width=3.35in]{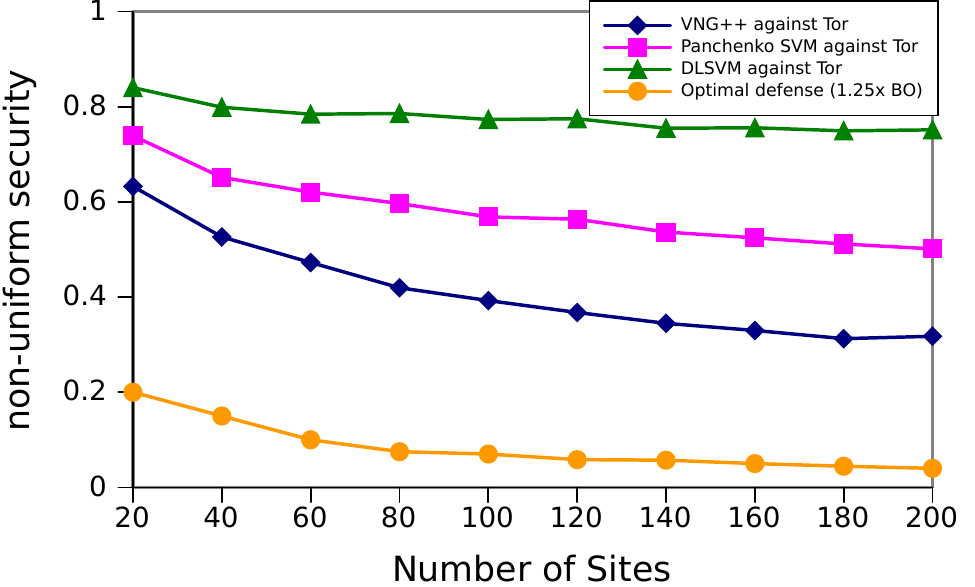}
    \label{fig:all-attacks-vs-sites-tor}
  }
  \subfigure[SSH]{
    \includegraphics[clip, trim = 1.1in 6.9in 3.5in 1.75in, width=3.35in]{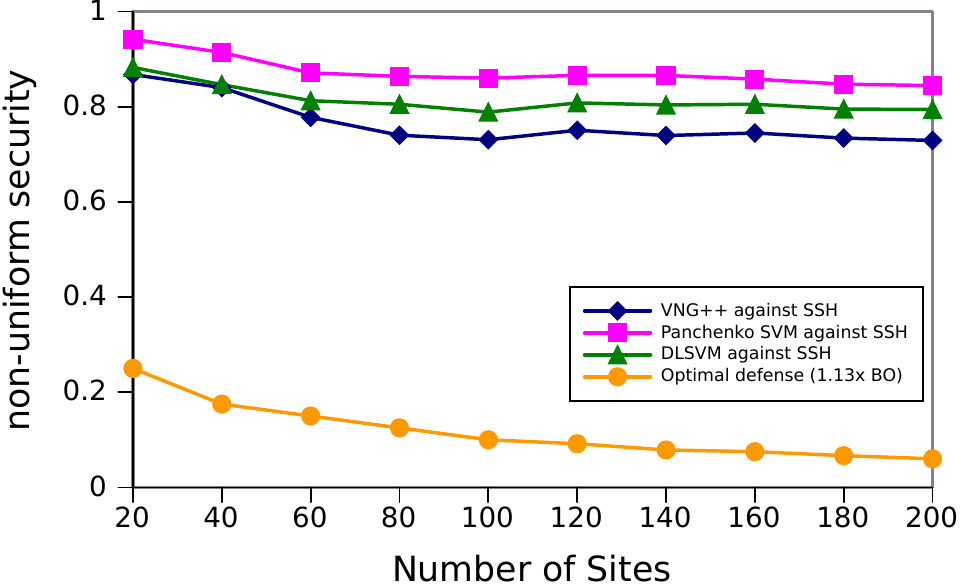}
    \label{fig:all-attacks-vs-sites-vanilla-ssh}
  }
  \caption{\label{fig:all-attacks-vs-sites-various}
  	Security of \csb, Tor, and SSH compared to the lower bounds from \Cref{sec:theory},
	  as a function of the number of possible web pages.
    }
\end{figure*}

\paragraph*{Bandwidth Cost} \Cref{fig:overhead} plots the bandwidth ratios
of SSH, Tor, and \csb with CTSP and CPSP padding.  SSH has almost no
overhead, and Tor's overhead is about 25\% on average.  \csb with CPSP
has an average overhead of 129\%, CTSP has average overhead 180\%.  Thus
\csb's improved security does come at a price.

\begin{figure}[t]
  \centering
  \includegraphics[clip, trim = 1.1in 6.9in 3.5in 1.75in, width=3.35in]{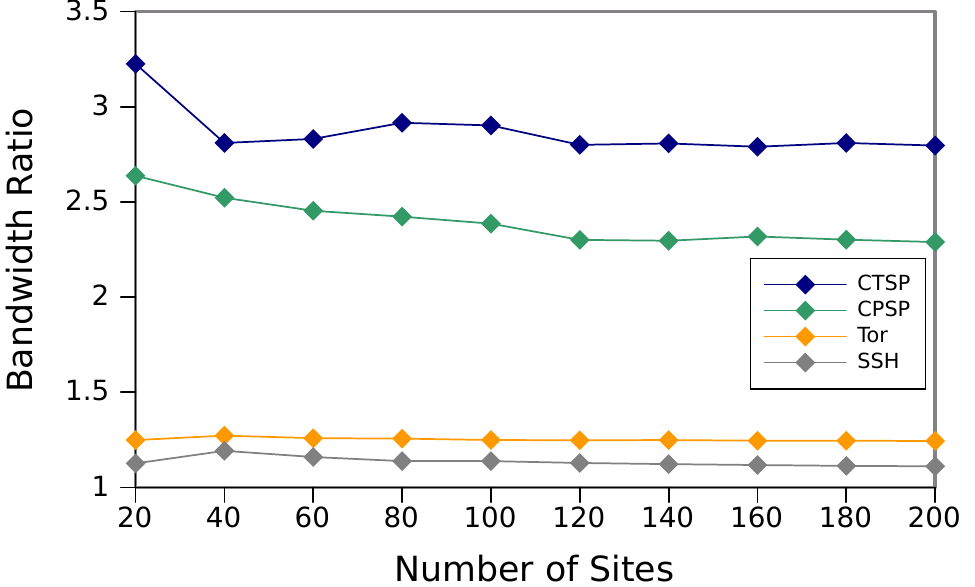}
  \caption{\label{fig:overhead} Bandwidth ratios of various defense schemes as a function of the number of possible web pages.}
\end{figure}

\paragraph*{Theoretical Bounds}
\Cref{fig:all-attacks-tradeoff-various} evaluates \csb, Tor, SSH, and
\buflo against the theoretical lower bounds developed in
\Cref{sec:theory}.

\Cref{fig:bounds_dlsvm} shows the results of our empirical evaluation
of \csb, Tor, and SSH on $n=120$ sites and using the DLSVM attack to
estimate security.  We limit to 120 sites to make it easier to compare
with the \buflo results reported by Dyer, et al., and which use 128
sites.  There is a significant gap between the bandwidth of \csb and
the lower bound.  However, as can be seen in
\Cref{fig:closeness_dlsvm}, \csb in CTSP mode is over 6$\times$ closer
to the trade-off lower bound than Tor for 200 sites, and is the most
efficient scheme across all sizes we measured.

\Cref{fig:bounds_panchenko} presents the results of our empirical
evaluation of \csb, Tor, and SSH on $n=120$ websites, busing the
Panchenko attack to estimate security.  We also present Dyer's
reported results from their experiments with \buflo on 128 sites, also
using the Panchenko attack.  Note that, since Dyer used 128 sites to
evaluate \buflo, this slightly over-estimates \buflo's security
compared to the other schemes plotted in the figure.  Also, recall
that Dyer's experiments with \buflo were all based on simulation.  

Despite the differences in experimental methodology, we can see that
\csb offers performance in the same general range as the \buflo
configurations from Dyer's paper, but has slightly worse security in
our experiments.  

\Cref{fig:closeness_panchenko} shows that, based on our experiments
and the simulation results of Dyer, et al., all but one \buflo
configuration get closer to the trade-off lower bound curve than \csb,
Tor, and SSH (SSH is omitted from the graph because its ratio to the
lower bound was never less than 400).  This figure also highlights a
difference between the DLSVM and Panchenko attacks.  In the DLSVM
results shown in \Cref{fig:closeness_dlsvm}, Tor and SSH diverge from
\csb.  In the Panchenko results in \Cref{fig:closeness_panchenko}, Tor
and \csb appear to be equally close to the lower bound.

\begin{figure*}[t]
  \centering
  
  \subfigure[DLSVM, $n=120$]{
    \includegraphics[clip, trim = 1.1in 6.9in 3.5in 1.75in, width=3.35in]{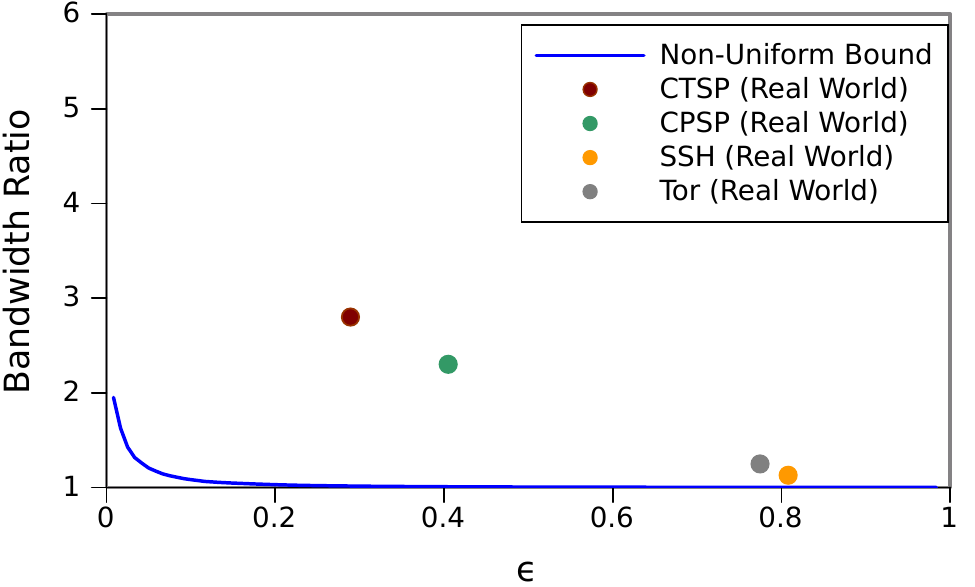}
    \label{fig:bounds_dlsvm}
  }
  \subfigure[Panchenko, $n=120$]{
    \includegraphics[clip, trim = 1.1in 6.9in 3.5in 1.75in, width=3.35in]{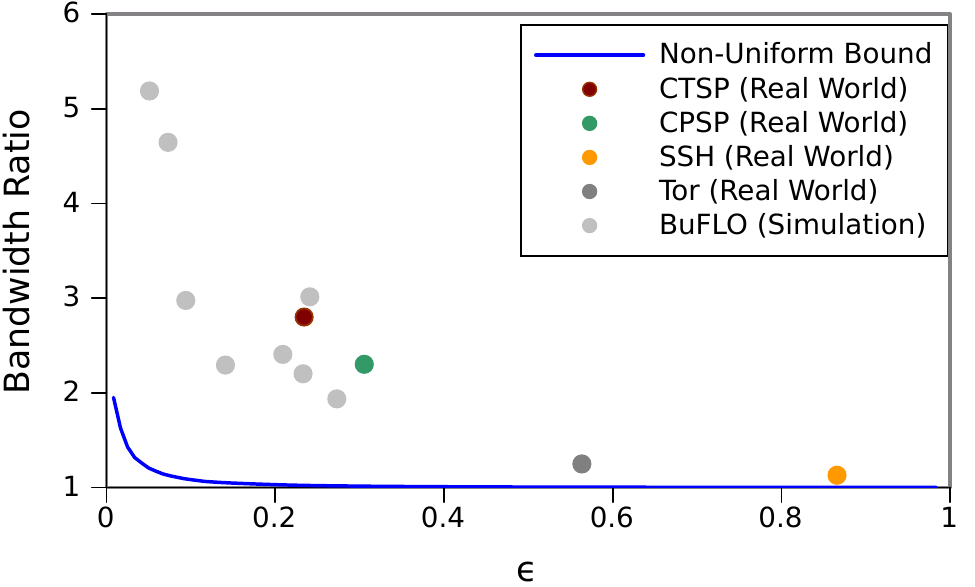}
    \label{fig:bounds_panchenko}
  }
  
  \subfigure[DLSVM]{
  \includegraphics[clip, trim = 1.1in 6.9in 3.5in 1.7in, width=3.35in]{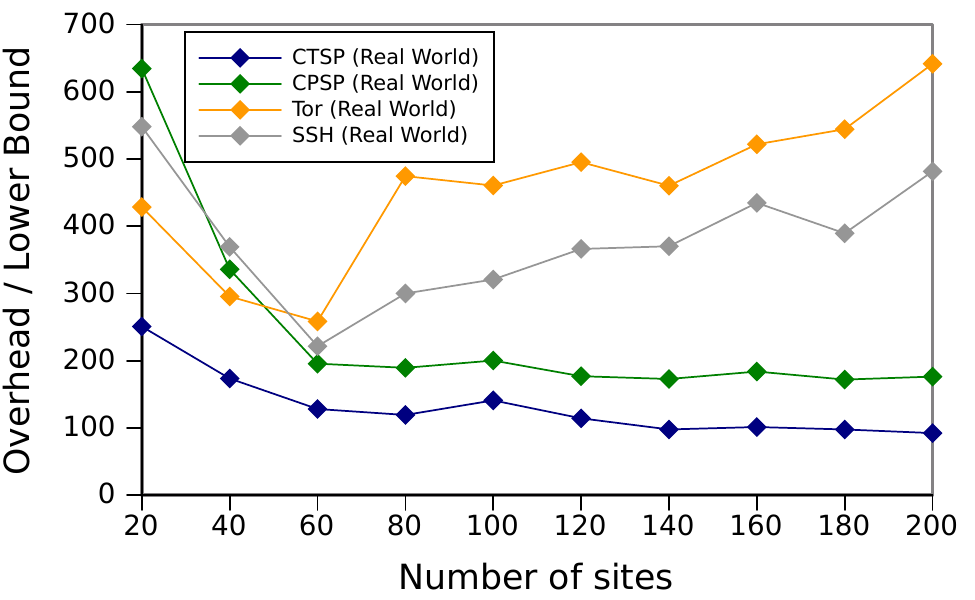}
    \label{fig:closeness_dlsvm}
  }
  \subfigure[Panchenko]{
    \includegraphics[clip, trim = 1.1in 6.9in 3.5in 1.7in, width=3.35in]{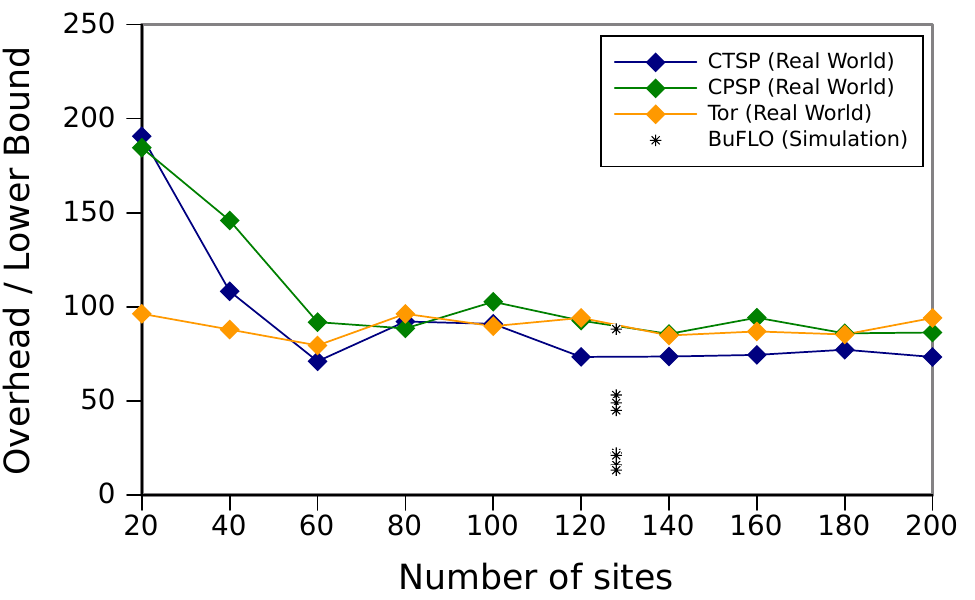}
    \label{fig:closeness_panchenko}
  }

  \caption{\label{fig:all-attacks-tradeoff-various} Non-uniform lower
    bounds on bandwidth ratio, as a function of the security
    parameter, $\epsilon$, and specific trade-off points of the
    systems evaluated.  The \buflo results are taken from Dyer, et
    al.~\cite{dyer-snp12}, and therefore use $n=128$.  SSH is omitted
    from \Cref{fig:closeness_panchenko} because its ratio to the lower
    bound was always greater than 400.}
\end{figure*}



\section{Discussion}
\label{sec:discussion}

Since early termination does not seem to affect security, the padding results
suggest that the padding performed while transmitting a website
sufficiently hides the size of the website, so that additional stream
padding at the end of the transmission has little security benefit.
Additional client padding does improve security, though -- probably by
obscuring the size of the final object requested by the client.

The lower bounds derived in \Cref{sec:theory} proved useful for
comparing schemes.  For example, without the lower bounds, it is
difficult to determine whether Tor, SSH, or \csb has the greatest
efficiency in \Cref{fig:bounds_dlsvm}, but it becomes obvious in
\Cref{fig:closeness_dlsvm}.  


Overall, \csb has better security than any other defense in our
experiments, albeit at greater expense.  It has the best
security/overhead trade-off, as well.

\csb's security/overhead trade-off is in the same range as the
estimates Dyer obtained for \buflo in their simulations.  For example,
Dyer, et al., reported that, in one configuration of \buflo, bandwidth
overhead was 200\% and the Panchenko SVM had an 24.1\% success rate on
128 websites.  We found that \csb with CTSP padding had an overhead of
180\% on 120 websites, and that the Panchenko SVM had a success rate
of 23.4\%.

\csb's congestion-sensitivity likely had little impact in these
experiments, which were carried out on a fast local network, so that
congestion was rare.  However, \csb's congestion-sensitivity means
that, in a real deployment, it would have even better bandwidth
overhead.

\csb's latency overhead is approximately 3 in all our experiments.
This is better than Tor's latency, although Tor has the additional
overhead of onion routing, so no fair comparison is possible.  We
cannot compare with the latency estimates reported by Dyer, et al.,
because they gave only absolute latency values.

\section{Conclusion}
\label{sec:conclusion}

\csbuflo offers a high-security, moderate-overhead solution to website
fingerprinting attacks.  Compared to SSH and Tor, it achieves a better
security/bandwidth trade-off, i.e. it uses its bandwidth efficiently
to provide extra security.  Our experiments also show that it has
acceptable latencies.  The padding schemes developed in this paper,
along with browser-coordination and early-termination algorithm, can
improve security with less overhead than previous stream padding
schemes.  Interestingly, we also found that padding from one end of a
connection can sometimes be an efficient way to hide information about
the data sent from the other side of the connection.


Our theoretical results provide new tools for comparing defense
systems.  More importantly, they suggest that a small amount of
well-placed cover traffic can make many websites look similar.
Therefore, the reason website fingerprinting defenses are so expensive
is not because websites are so different.  Rather, it is because the
defense, operating blindly, does not know where to put the cover
traffic, and so it must put it everywhere.  An interesting direction
for future research is to attempt to approximate the knowledge of an
offline defense by having a real defense remember information about
websites seen in the past.


\bibliographystyle{plain}

\bibliography{rob}

\end{document}